\documentclass[submission,copyright,creativecommons]{eptcs}
\providecommand{\event}{WPTE 2016} 
\usepackage{breakurl}             
\usepackage{underscore}           

\title{The Complexity of Abstract Machines}
\author{Beniamino Accattoli
\institute{INRIA \& LIX, \'Ecole Polytechnique}
\email{beniamino.accattoli@inria.fr}
}
\def\titlerunning{The Complexity of Abstract Machines}
\def\authorrunning{B. Accattoli}

\newcommand{\macrospath}{.}

\usepackage{ifthen}

\newboolean{talk}
\setboolean{talk}{false}
\newboolean{paper}
\setboolean{paper}{false}

\newboolean{IEEEstyle}
\setboolean{IEEEstyle}{false}
\newboolean{lipicsstyle}
\setboolean{lipicsstyle}{false}
\newboolean{eptcsstyle}
\setboolean{eptcsstyle}{false}
\newboolean{sigplanstyle}
\setboolean{sigplanstyle}{false}
\newboolean{easychairstyle}
\setboolean{sigplanstyle}{false}

\newboolean{needstheorems}
\setboolean{needstheorems}{false}
\newboolean{withimages}
\setboolean{withimages}{false}
\newboolean{withproofs}
\setboolean{withproofs}{true}

\newboolean{french}
\setboolean{french}{false}

\setboolean{eptcsstyle}{true}
\ifthenelse{\boolean{talk}}{}{\usepackage[utf8]{inputenc}}
\ifthenelse{\boolean{french}}{\usepackage[french]{babel}}{
  \ifthenelse{\boolean{lipicsstyle}}{}{\usepackage[english]{babel}}}
\usepackage{amssymb}
\usepackage{amsmath}
\usepackage{graphicx}
\usepackage{bussproofs}
\usepackage{cmll}
\usepackage{stmaryrd}
\usepackage{multirow}
\ifthenelse{\boolean{talk}}{\usepackage{color}}{
  \ifthenelse{\boolean{lipicsstyle}}{\usepackage{color}}{\usepackage[usenames]{xcolor}}}
\ifthenelse{\boolean{IEEEstyle}}{
\usepackage[bookmarks={false}]{hyperref}}{
\usepackage{hyperref}}
\usepackage{fancybox}
\usepackage{xspace}

\ifthenelse{\boolean{IEEEstyle}}{
\setboolean{needstheorems}{true}}{}

\ifthenelse{\boolean{eptcsstyle}}{
\setboolean{needstheorems}{true}}{}

\ifthenelse{\boolean{sigplanstyle}}{
\setboolean{needstheorems}{true}}{}

\ifthenelse{\boolean{needstheorems}}{

\usepackage{amsthm}

    \newtheorem{theorem}{Theorem}[section]
    \newtheorem{lemma}[theorem]{Lemma}

    \newtheorem{proposition}[theorem]{Proposition}
    \newtheorem{definition}[theorem]{Definition}

}{}

\newcommand{\myproof}[1]{
\ifthenelse{\boolean{withproofs}}{#1}{}
}

\newcommand{\la}[1]{\lambda #1.}

\newcommand{\tm}{t}
\newcommand{\tmtwo}{s}
\newcommand{\tmthree}{u}
\newcommand{\tmfour}{r}

\newcommand{\var}{x}
\newcommand{\vartwo}{y}
\newcommand{\varthree}{z}

\newcommand{\nfpar}[2]{{\tt nf}_{#2}(#1)}

\newcommand{\Rew}[1]{\rightarrow_{#1}}

\renewcommand{\to}{\Rew{}}
\newcommand{\tob}{\Rew{\beta}}

\newcommand{\towh}{\Rew{wh}}

\newcommand{\tostrat}{\Rew{x}}

\newcommand{\csym}{{\tt c}}

\newcommand{\ctxholep}[1]{\langle #1\rangle}
\newcommand{\ctxhole}{\ctxholep{\cdot}}

\newcommand{\evctx}{E}

\newcommand{\evctxp}[1]{\evctx\ctxholep{#1}}

\newcommand{\nbvctxtwo}[1]{\nbvctxtwo{#1}}

\newcommand{\defeq}{:=}

\newcommand{\isub}[2]{\{#1/#2\}}
\newcommand{\esub}[2]{[#1/#2]}
\renewcommand{\esub}[2]{[#1{\shortleftarrow}#2]}
\renewcommand{\isub}[2]{\{#1{\shortleftarrow}#2\}}

\newcommand{\llbrace}{\{ \kern -0.27em \vert}
\newcommand{\rrbrace}{\vert \kern -0.27em \}}

\renewcommand{\l}{\lambda}
\newcommand{\ie}{{\em i.e.}\xspace}

\newcommand{\ih}{{\em i.h.}\xspace}

\newcommand{\ignore}[1]{}

\newcommand{\myinput}[1]{\ifthenelse{\boolean{withimages}}{\input{#1}}{}}

\newcommand{\reflemma}[1]{Lemma~\ref{l:#1}}
\newcommand{\reflemmap}[2]{Lemma~\ref{l:#1}.\ref{p:#1-#2}}

\newcommand{\refprop}[1]{Proposition~\ref{prop:#1}}
\newcommand{\refsect}[1]{Sect.~\ref{sect:#1}}

\ifthenelse{\boolean{easychairstyle}}{}{
	\newcommand{\refeq}[1]{(\ref{eq:#1})}
}
\newcommand{\reffig}[1]{Fig.~\ref{fig:#1}}

\newcommand{\refdef}[1]{Definition~\ref{def:#1}}

\newcommand{\refpoint}[1]{Point~\ref{p:#1}}

\newcommand{\nat}{\mathbb{N}}

\newcommand{\size}[1]{|#1|}
\newcommand{\sizep}[2]{|#1|_{#2}}

\newcommand{\code}{\overline{\tm}}
\newcommand{\codetwo}{\overline{\tmtwo}}
\newcommand{\codethree}{\overline{\tmthree}}

\newcommand{\clos}{c}
\newcommand{\clostwo}{\clos'}

\newcommand{\env}{e}
\newcommand{\envtwo}{\env'}

\newcommand{\genv}{E}
\newcommand{\genvtwo}{\genv'}
\newcommand{\genvthree}{\genv''}

\newcommand{\stack}{\pi}

\newcommand{\state}{s}
\newcommand{\statetwo}{s'}

\newcommand{\exec}{\rho}
\newcommand{\exectwo}{\sigma}

\newcommand{\decode}[1]{\underline{#1}}
\newcommand{\decodep}[2]{\decode{#1}\ctxholep{#2}}

\newcommand{\stempty}{\epsilon}
\newcommand{\cons}{::}

\newcommand{\deriv}{\ensuremath{d}}

\newcommand{\rename}[1]{#1^\alpha}

\newcommand{\mach}{{\tt M}}

\newcommand{\admsym}{\csym}

\newcommand{\tomachhole}[1]{\leadsto_{#1}}
\newcommand{\tomach}{\tomachhole{}}

\newcommand{\tomachmone}{$\begin{array}{c|c|ccc|c|c}
		Code & Stack &  Env && Code & Stack & Env\tomachhole{\mulsym_1}}

\newcommand{\tomachc}{\to$\begin{array}{c|c|ccc|c|c}
		Code & Stack &  Env && Code & Stack & Envmachhole{\admsym}}

\newcommand{\statesearch}[2]{(#1,#2)}
\newcommand{\statemicro}[2]{(#1,#2)}
\newcommand{\statemam}[3]{(#1,#2,#3)}

\newcommand{\statekamclos}[2]{(#1,#2)}

\newcommand{\compil}[1]{#1^\circ}
\newcommand{\applsym}{@l}
\newcommand{\rootbetasym}{r\beta}
\newcommand{\varsym}{var}
\newcommand{\tomachine}{\tomachhole\mach}
\newcommand{\tomachrbeta}{\tomachhole{\rootbetasym}}
\newcommand{\tomachappl}{\tomachhole{\applsym}}
\newcommand{\tomachvar}{\tomachhole{\varsym}}

\newcommand{\sizepr}[1]{\sizep{#1}{\beta}}
\newcommand{\sizebeta}[1]{\sizepr{#1}}
\newcommand{\sizeappl}[1]{\sizep{#1}{\applsym}}
\newcommand{\sizevar}[1]{\sizep{#1}{\varsym}}

\renewcommand{\tmtwo}{u}
\renewcommand{\tmthree}{r}
\renewcommand{\tmfour}{p}

\usepackage{hhline}

\begin{document}
\maketitle

\renewcommand{\tostrat}{\to}

\begin{abstract}
The $\l$-calculus is a peculiar computational model whose definition does not come with a notion of machine. Unsurprisingly, implementations of the $\l$-calculus have been studied for decades. Abstract machines are implementations schema for fixed evaluation
strategies that are a compromise between theory and practice: they are concrete enough to provide a notion of machine and abstract enough to avoid the many intricacies of actual implementations. There is an extensive literature about abstract machines for the $\l$-calculus, and yet---quite mysteriously---the efficiency of these machines with respect to the strategy that they implement has almost never been studied.

This paper provides an unusual introduction to abstract machines, based on the complexity of their overhead with respect to the length of the implemented strategies. It is conceived to be a tutorial, focusing on the case study of implementing the weak head (call-by-name) strategy, and yet it is an original re-elaboration of known results. Moreover, some of the observation contained here never appeared in print before.
\end{abstract}

\section{Cost Models \& Size-Explosion}
The $\l$-calculus is an undeniably elegant computational model. Its definition is given by three constructors and only one computational rule, and yet it is Turing-complete. A charming feature is that it does not rest on any notion of machine or automaton. The catch, however, is that its cost model are far from being evident. What should be taken as time and space measures for the $\l$-calculus? The natural answers are the number of computational steps (for time) and the maximum size of the terms involved in a computation (for space). Everyone having played with the $\l$-calculus would immediately point out a problem: the $\l$-calculus is a nondeterministic system where the number of steps depends much on the evaluation strategy, so much that some strategies may diverge when others provide a result (but fortunately the result, if any, does not depend on the strategy). While this is certainly an issue to address, it is not the serious one. The big deal is called \emph{size-explosion}, and it affects all evaluation strategies.

\paragraph{Size-Explosion.} There are families of terms where the size of the $n$-th term is linear in $n$, evaluation takes a linear number of steps, but the size of the result is exponential in $n$. Therefore, the number of steps does not even account for the time to write down the result, and thus at first sight it does not look as a reasonable cost model. Let's see examples.

The simplest one is a variation over the famous looping $\l$-term $\Omega \defeq (\la\var \var\var) (\la\var \var\var) \tob \Omega \tob \ldots$. In $\Omega$ there is an infinite sequence of duplications. In the first size-exploding family there is a sequence of $n$ nested duplications. We define both the family $\{\tm_n\}_{n\in\nat}$ of size-exploding terms and the family $\{\tmtwo_n\}_{n\in\nat}$ of results of the evaluation

\begin{center}
$\begin{array}{cccccccc}
  \tm_0 & \defeq & \vartwo &&& \tmtwo_0 & \defeq & \vartwo\\
  \tm_{n+1} & \defeq & (\la\var \var\var) \tm_n &&&   \tmtwo_{n+1} & \defeq & \tmtwo_n \tmtwo_n
\end{array}$  
\end{center}

We use $\size\tm$ for the size of a term, \ie the number of symbols to write it, and say that a term is \emph{neutral} if it is normal and it is not an abstraction. 

\begin{proposition}[Open and Rightmost-Innermost Size-Explosion]
Let $n\in \nat$. Then $\tm_n \tob^n \tmtwo_n$, moreover $\size{\tm_n} = O(n)$, $\size{\tmtwo_n} = \Omega(2^n)$, and $\tmtwo_n$ is neutral.
\end{proposition}

\begin{proof}
By induction on $n$. The base case is immediate. The inductive case: $\tm_{n+1} = (\la\var \var\var) \tm_n \tob^n (\la\var \var\var) \tmtwo_n \tob \tmtwo_n \tmtwo_n = \tmtwo_{n+1}$, where the first sequence is obtained by the \ih The bounds on the sizes are immediate, as well as the fact that $\tmtwo_{n+1}$ is neutral.
\end{proof}

\paragraph{Strategy-Independent Size-Explosion.} The example relies on rightmost-innermost evaluation (\ie the strategy that repeatedly selects the rightmost-innermost $\beta$-redex) and open terms (the free variable $\tm_0 = \vartwo$). In fact, evaluating the same family in a leftmost-outermost way would produce an exponentially long evaluation sequence. One may then believe that size-explosion is a by-product of a clumsy choice for the evaluation strategy. Unfortunately, this is not the case. It is not hard to modify the example as to make it strategy-independent, and it is also easy to get rid of open terms. Let the identity combinator be $I \defeq \la\varthree \varthree$ (it can in fact be replaced by any closed abstraction). Define
\begin{align*}
\tmthree_1 & \defeq \la{\var}\la{\vartwo}(\vartwo \var \var)  & \tmfour_0 & \defeq I \\
\tmthree_{n+1} & \defeq \la{\var}(\tmthree_n (\la{\vartwo}(\vartwo \var \var))) & \tmfour_{n+1} & \defeq \la{\vartwo}(\vartwo \tmfour_{n} \tmfour_{n})	
\end{align*}

The size-exploding family is $\{\tmthree_n I\}_{n\in\nat}$, \ie it is obtained by applying $\tmthree_n$ to the identity $I = \tmfour_0$. The statement we are going to prove is in fact more general, about $\tmthree_n \tmfour_m$ instead of just $\tmthree_n I$, in order to obtain a simple inductive proof.

\begin{proposition}[Closed and Strategy-Independent Size-Explosion]
\label{prop:abs-size-explosion}
  Let $n \!>\! 0$. Then $\tmthree_n \tmfour_m \tob^n  \tmfour_{n+m}$, and in particular $\tmthree_n I \tob^n  \tmfour_n$. Moreover, $\size{\tmthree_n I} = O(n)$, $\size{\tmfour_n} = \Omega(2^n)$, $\tmthree_n I$ is closed, and $\tmfour_n$ is normal.
\end{proposition}

\begin{proof}
By induction on $n$. The base case: $\tmthree_1 \tmfour_m = \la{\var}\la{\vartwo}(\vartwo \var \var) \tmfour_m \tob (\la{\vartwo}(\vartwo \tmfour_m \tmfour_m)) = \tmfour_{m+1}$. The inductive case: $\tmthree_{n+1} \tmfour_m = \la{\var}(\tmthree_n (\la{\vartwo}(\vartwo \var \var))) \tmfour_m  \tob \tmthree_n (\la{\vartwo}(\vartwo \tmfour_m \tmfour_m)) = \tmthree_n \tmfour_{m+1} \tob^n \tmfour_{n+m+1}$, where the second sequence is obtained by the \ih The rest of the statement is immediate.
\end{proof}

The family $\{\tmthree_n I\}_{n\in\nat}$ is interesting because no matter how one looks at it, it always explodes: if evaluation is weak (\ie it does not go under abstraction) there is only one possible derivation to normal form and if it is strong (\ie unrestricted) all derivations have the same length (and are permutatively equivalent). To our knowledge this family never appeared in print before.
\section{\texorpdfstring{The $\l$-Calculus is Reasonable, Indeed}{The Lambda-Calculus is Reasonable, Indeed}}
\label{sect:indeed}
Surprisingly, the isolation and the systematic study of the size-explosion problem is quite recent---there is no trace of it in the classic books on the $\l$-calculus, nor in any course notes we are aware of. Its essence, nonetheless, has been widespread folklore for a long time: in practice, functional languages never implement full $\beta$-reduction, considered a costly operation, and theoretically the  $\l$-calculus is usually considered a model not suited for complexity analyses.


A way out of the issue of cost models for the $\l$-calculus, at first sight, is to take the time and space required for the execution of a $\l$-term in a fixed implementation. There is however no canonical implementation. The design of an implementation in fact rests on a number of choices. Consequently, there are a number of different but more or less equivalent machines taking a different number of steps and using different amounts of space to evaluate a term. Fixing one of them would be arbitrary, and, most importantly, would betray the machine-independent spirit of the $\l$-calculus.

\paragraph{Micro-Step Operational Semantics.} Luckily, the size-explosion problem can be solved in a machine-independent way. Somewhat counterintuitively, in fact, the number of $\beta$-steps can be taken as a reasonable cost model. The basic idea is simple: one has to step out of the $\l$-calculus, by switching to a different setting that \emph{mimics} $\beta$-reduction without literally doing it, acting on \emph{compact representations} of terms to avoid size-explosion. Essentially, the recipe requires four ingredients:
\begin{enumerate}
  \item \emph{Statics}: $\l$-terms are refined with a form of \emph{sharing} of subterms;
  \item \emph{Dynamics}: evaluation has to manipulate terms with sharing via \emph{micro}-operations;
  \item \emph{Cost}: these micro-step operations have constant cost;
  \item \emph{Result}: micro-evaluation stops on a \emph{shared representation of the result}.
\end{enumerate}
The recipe leaves also some space for improvisation: $\l$-calculus can in fact be enriched with \emph{first-class sharing} in various ways. Mainly, there are three approaches: \emph{abstract machines}, \emph{explicit substitutions}, and \emph{graph rewriting}. They differ in the details but not in the essence---they can be grouped together under the slogan \emph{micro-step operational semantics}.

\paragraph{Reasonable Strategies.} An evaluation strategy $\tostrat$ for the $\l$-calculus is \emph{reasonable} if there is a micro-step operational semantics $M$ mimicking $\tostrat$ and such that the number of micro-steps to evaluate a term $\tm$ is polynomial in the number of $\tostrat$-steps to evaluate $\tm$ (and in the size of $\tm$, we will come back to this point later on). If a strategy $\tostrat$ is reasonable then its length is a reasonable cost model, despite size-explosion: the idea is that the $\l$-calculus is kept as an \emph{abstract} model, easy to define and reason about, while complexity-concerned evaluation is meant to be performed at the more sophisticated micro-step level, where the explosion cannot happen.

Of course, the design of a reasonable micro-step operational semantics depends much on the strategy and the chosen flavor of micro-steps semantics, and it can be far from easy. For \emph{weak} strategies---used to model functional programming languages---reasonable micro-steps semantics are based on a simple form of sharing. The first result about reasonable strategies was obtained by Blelloch and Greiner in 1995 \cite{DBLP:conf/fpca/BlellochG95} and concerns indeed a weak strategy, namely the call-by-value one. At the micro-step level it relies on abstract machines. Similar results were then proved again, independently, by Sands, Gustavsson, and Moran in 2002 \cite{DBLP:conf/birthday/SandsGM02} and by Dal Lago and Martini in 2006 \cite{DBLP:conf/cie/LagoM06}. For \emph{strong} strategies---at work in proof assistant engines---quite more effort and care are required. A sophisticated second-level of sharing, called \emph{useful sharing}, is necessary to obtain reasonable micro-step semantics for strong evaluation. The first such semantics has been introduced by Accattoli and Dal Lago in 2014 \cite{DBLP:conf/csl/AccattoliL14} for the leftmost-outermost strategy, and its study is still ongoing \cite{DBLP:conf/lics/AccattoliC15,DBLP:conf/wollic/Accattoli16}.

\paragraph{The Complexity of Abstract Machines.} To sum up, various techniques, among which abstract machines, can be used to prove that the number of $\beta$-steps is a reasonable time cost model, \ie a metric for time complexity. The study can then be reversed, exploring how to use this metric to study the relative complexity of abstract machines, that is, the complexity of the overhead of the machine with respect to the number of $\beta$-steps. Such a study leads to a new quantitative theory of abstract machines, where machines can be compared and the value of different design choices can be measured. The rest of the paper provides a gentle introduction to the basic concepts of the new complexity-aware theory of abstract machines being developed by the author in joint works \cite{DBLP:conf/icfp/AccattoliBM14,DBLP:conf/wollic/AccattoliC14,DBLP:conf/aplas/AccattoliBM15,DBLP:conf/lics/AccattoliC15,DBLP:conf/wollic/Accattoli16} with Damiano Mazza, Pablo Barenbaum, and Claudio Sacerdoti Coen, and resting on tools and concepts developed beforehand in collaborations with Delia Kesner \cite{DBLP:conf/csl/AccattoliK10} and Ugo Dal Lago \cite{DBLP:conf/rta/AccattoliL12}, as well as Kesner plus Eduardo Bonelli and Carlos Lombardi \cite{DBLP:conf/popl/AccattoliBKL14}.

\paragraph{Case Study: Weak Head Strategy.} The paper focuses on a case study, the weak head (call-by-name) strategy, also known as weak head reduction (we use \emph{reduction} and \emph{strategy} as synonymous, and prefer \emph{strategy}), and defined as follows:
 \begin{equation}
\begin{array}{c@{\hspace{2em}}ccccc}
		\AxiomC{}
	\RightLabel{{\tiny (root $\beta$)}}
	\UnaryInfC{$(\la\var\tm)\tmtwo \towh \tm\isub\var\tmtwo$}
	\DisplayProof
	&
	\AxiomC{$\tm \towh \tmtwo$}
	\RightLabel{{\scriptsize (@l)}}
	\UnaryInfC{$\tm\tmthree \towh \tmtwo\tmthree$}
	\DisplayProof
\end{array}
\label{eq:cbn-strategy}
\end{equation}
This is probably the simplest possible evaluation strategy. Of course, it is deterministic. Let us mention two other ways of defining it, as they will be useful in the sequel. First, the given inductive definition can be unfolded into a single synthetic rule $(\la\var\tm)\tmtwo \tmthree_1 \ldots \tmthree_k\towh \tm\isub\var\tmtwo \tmthree_1 \ldots \tmthree_k$. Second, the strategy can be given via evaluation contexts: define $\evctx \defeq \ctxhole \mid \evctx \tmthree$ and define $\towh$ as $\evctxp{(\la\var\tm)\tmtwo} \towh \evctxp{\tm\isub\var\tmtwo}$ (where $\evctxp\tm$ is the operation of plugging $\tm$ in the context $\evctx$, consisting in replacing the hole $\ctxhole$ with $\tm$).

Sometimes, to stress the modularity of the reasoning, we will abstract the weak head strategy into a generic strategy $\tostrat$. Last, a \emph{derivation} is a possibly empty sequence of rewriting steps.
\section{Introducing Abstract Machines}
\label{sect:machines-intro}
\paragraph{Tasks of Abstract Machines.} An abstract machine is an implementation schema for an evaluation strategy $\tostrat$ with sufficiently atomic operations and without too many details. A machine for $\tostrat$ accounts for 3 tasks:
\begin{enumerate}
  \item \emph{Search}: searching for $\tostrat$-redexes;
  \item \emph{Substitution}: replace meta-level substitution with an approximation based on sharing;
  \item \emph{Names}: take care of $\alpha$-equivalence.
\end{enumerate}

\paragraph{Dissecting Abstract Machines.} To guide the reader through the different concepts to design and analyze abstract machines, the next two subsections describe in detail two toy machines that address in isolation the first two mentioned tasks, \emph{search} and \emph{substitution}. They will then be merged into the Milner Abstract Machine (MAM). In \refsect{intro-complexity} we will analyze the complexity of the MAM. Next, we will address \emph{names} and describe the Krivine Abstract Machine, and quickly study its complexity. 

\paragraph{Abstract Machines Glossary.} 
\begin{itemize}
\item An abstract machine $\mach$ is given by \emph{states}, noted $\state$, and \emph{transitions} between them, noted $\tomach$;
\item A state is given by the \emph{code under evaluation} plus some \emph{data-structures} to implement \emph{search} and \emph{substitution}, and to take care of \emph{names};
\item The code under evaluation, as well as the other pieces of code scattered in the data-structures, are $\l$-terms \emph{not considered modulo $\alpha$-equivalence};
\item Codes are over-lined, to stress the different treatment of $\alpha$-equivalence;
\item A code $\code$ is \emph{well-named} if $\var$ may occur only in $\codetwo$ (if at all) for every sub-code $\la\var\codetwo$ of $\code$;
\item A state $\state$  is \emph{initial} if its code is well-named and its data-structures are empty; 
\item Therefore, there is a bijection $\compil\cdot$ (up to $\alpha$) between terms and initial states, called \emph{compilation}, sending a term $\tm$ on the initial state $\compil\tm$ on a well-named code $\alpha$-equivalent to $\tm$;
\item An \emph{execution} is a (potentially empty) sequence of transitions $\statetwo \tomach^* \state$ from an initial state $\statetwo$ obtained by compiling a(n initial) term $\tm_0$;
\item A state $\state$  is \emph{reachable} if it can be obtained as the end state of an execution;
\item A state $\state$ is \emph{final} if it is reachable and no transitions apply to $\state$.
\item A machine comes with a map $\decode\cdot$ from states to terms, called \emph{decoding}, that on initial states is the inverse (up to $\alpha$) of compilation;

\item A machine $\mach$ has a set of \emph{$\beta$-transitions} that are meant to be mapped to $\beta$-redexes (and whose name involves $\beta$) by the decoding, while the remaining \emph{overhead transitions} are mapped on equalities;
\item We use $\size\exec$ for the length of an execution $\exec$, and $\sizepr\exec$ for the number of $\beta$-transitions in $\exec$.
\end{itemize}

\paragraph{Implementations.} For every machine one has to prove that it correctly implements the strategy it was conceived for. Our notion, tuned towards complexity analyses, requires a perfect match between the number of $\beta$-steps of the strategy and the number of $\beta$-transitions of the machine execution.

\begin{definition}[Machine Implementation]
\label{def:implem}
A machine $\mach$ \emph{implements a strategy} $\tostrat$ on $\lambda$-terms when given a $\l$-term $\tm$ the following holds
\begin{enumerate}
\item \emph{Executions to Derivations}: for any $\mach$-execution $\exec: \compil\tm \tomachine^* \state$ there exists a $\tostrat$-derivation $\deriv: \tm \tostrat^* \decode\state$.

\item \emph{Derivations to Executions}: for every $\tostrat$-derivation $\deriv: \tm \towh^* \tmtwo$ there exists a $\mach$-execution $\exec: \compil\tm \tomachine^* \state$ such that $\decode\state = \tmtwo$.

\item \emph{$\beta$-Matching}: in both previous points the number $\sizepr\exec$ of $\beta$-transitions in $\exec$ is exactly the length $\size\deriv$ of the derivation $\deriv$, \ie $\size\deriv = \sizepr\exec$.
\end{enumerate}
\end{definition}

Note that if a machine implements a strategy than the two are \emph{weakly bisimilar}, where weakness is given by the fact that overhead transitions do not have an equivalent on the calculus (hence their name). Let us point out, moreover, that the $\beta$-matching requirement in our notion of implementation is unusual but perfectly reasonable, as all abstract machines we are aware of do satisfy it.
\section{The Searching Abstract Machine}
Strategies are usually specified through inductive rules as those in \refeq{cbn-strategy}. The inductive rules incorporate in the definition the search for the next redex to reduce. Abstract machines make such a search explicit and actually ensure two related subtasks:
\begin{enumerate}
  \item Store the current evaluation context in appropriate \emph{data-structures}.
  \item Search \emph{incrementally}, exploiting previous searches.
\end{enumerate}
For weak head reduction the search mechanism is basic. The data structure is simply a stack $\stack$ storing the arguments of the current head subterm. 
 
\paragraph{Searching Abstract Machine.} The searching abstract machine (Searching AM) in \reffig{searching-am} has two components, the \emph{code} in evaluation position and the \emph{argument stack}.
 The machine has only two transitions, corresponding to the rules in \refeq{cbn-strategy}, one $\beta$-transition ($\tomachhole{r\beta}$) dealing with $\beta$-redexes in evaluation position and one overhead transition ($\tomachhole{@l}$) adding a term on the argument stack.
 \begin{figure}

   \begin{center}
   \begin{tabular}{c}
   \begin{tabular}{c|c}
   $\begin{array}{rrclc}
     \mbox{Stacks} & \stack & \defeq & \stempty \mid \code \cons \stack\\
     \mbox{Compilation} & \compil\tm & \defeq & \statesearch \code \stempty\\

   \end{array}$
   &
   $\begin{array}{rrclc}
     \mbox{Decoding} & \decode{ \statesearch \code \stempty } & \defeq & \tm\\
      & \decode{ \statesearch \code {\codetwo \cons \stack} } & \defeq & \decode{ \statesearch {\code \codetwo} \stack }\\

   \end{array}$
   \end{tabular}
\\\\
   $\begin{array}{|c|c||c||c|c|}
		\hline
		Code & Stack&Trans.& Code & Stack
		
		\\
		\hhline{=|=||=||=|=}		
		\code\codetwo & \stack
	  	&\tomachappl&
	  	\code & \codetwo\cons\stack\\
		
		\hhline{-|-||-||-|-}
						
		\la\var\code & \codetwo\cons\stack
	  	&\tomachrbeta &
		\code\isub\var{\codetwo} & \stack\\
		\hline
  	\end{array}
	$
   \end{tabular}   
  \end{center}
  \caption{Searching Abstract Machine (Searching AM).}
     \label{fig:searching-am}
\end{figure}
Compilation of a (well-named) term $\tm$ into a machine state simply sends $\tm$ to the \emph{initial state} $\statesearch\code\stempty$. The decoding given in \reffig{searching-am} is defined inductively on the structure of states. It can equivalently be given  contextually, by associating an evaluation context to the data structures---in our case sending the argument stack $\stack$ to a context $\decode\stack$ by setting $\decode{\stempty }  \defeq  \ctxhole$, $\decode{ \codetwo \cons \stack}  \defeq  \decode{ \stack}\ctxholep{\ctxhole \tmtwo}$,  and $\decode{ \statesearch \code \stack }  \defeq  \decodep \stack \tm$. It is useful to have both definitions since sometimes one is more convenient than the other.

\paragraph{Implementation.} We now show the implementation theorem for the Searching AM with respect to the weak head strategy. Despite the simplicity of the machine, we provide a quite accurate account of the proof of the theorem, to be taken as a modular recipe. The proofs of the other implementation theorems in the paper will then be omitted as they follow exactly the same structure, \emph{mutatis mutandis}. 

The \emph{executions-to-derivations} part of the implementation theorem always rests on a lemma about the decoding of transitions, that in our case takes the following form.

\begin{lemma}[Transitions Decoding]
\label{l:searching-am-one-step} 
Let $\state$ be a Searching AM state.
\begin{enumerate}
\item \label{p:searching-am-one-step-principal}
\emph{$\beta$-Transition}: if $\state \tomachrbeta \statetwo$ then $\decode\state \tob \decode\statetwo$.
\item \label{p:searching-am-one-step-commutative}
\emph{Overhead Transition}: if $\state \tomachappl \statetwo$ then $\decode\state = \decode\statetwo$.
\end{enumerate}
\end{lemma}

%
%

\proof
The first point is more easily proved using the contextual definition of decoding.
\begin{enumerate}
\item $\decode \state = \decode{ \statesearch{ \la\var\code}{ \codetwo\cons\stack } } = \decodep{ \codetwo\cons\stack }{ \la\var\tm } = \decodep{ \stack }{ (\la\var\tm) \tmtwo } \tob \decodep{ \stack }{ \tm \isub\var\tmtwo } = \decode\statetwo$.
\item $\decode \statetwo = \decode{ \statesearch{ \code }{ \codetwo\cons\stack } }=  \decode{ \statesearch{ \code\codetwo }{ \stack } } = \decode{ \state }$.\qed
\end{enumerate}\medskip

Transitions decoding extends to a projection of executions to derivations (via a straightforward induction on the length of the execution), as required by the implementation theorem. For the \emph{derivations-to-executions} part of the theorem, we proceed similarly, by first proving that single weak head steps are simulated by the Searching AM and then extending the simulation to derivations via an easy induction. There is a subtlety, however, because, if done naively, one-step simulations do not compose.

Let us explain the point. Given a step $\tm \towh \tmtwo$ there exists a state $\state$ such that $\compil\tm \tomachappl^* \tomachrbeta \state$ and $\decode\state = \tmtwo$, as expected. This property, however, cannot be iterated to build a many-steps simulation, because $\decode\state = \tmtwo$ does not imply $\state = \compil\tmtwo$, \ie $\state$  in general is not the compilation of $\tmtwo$. To make things work, the simulation of $\tm \towh \tmtwo$ should not start from $\compil\tm$ but from a state $\statetwo$ such that $\decode\statetwo = \tm$. Now, the proof of the step simulation lemma we just described relies on the following three properties:

\begin{lemma}[Bricks for Step Simulation]\hfill
\label{l:searching-am-bricks} 
\begin{enumerate}
	\item \label{p:searching-am-bricks-termination}
	\emph{Vanishing Transitions Terminate}:  $\tomachappl$ terminates;

	\item \label{p:searching-am-bricks-determinism}
	\emph{Determinism}: the Searching AM is deterministic;

	\item \label{p:searching-am-bricks-progress}
	\emph{Progress}: final Searching AM states decode to $\towh$-normal terms.
\end{enumerate} 
\end{lemma}

\begin{proof}
	\emph{Termination}: $\tomachappl$-sequences are bound by the size of the code. \emph{Determinism}: $\tomachrbeta$ and $\tomachappl$ clearly do not overlap and can be applied in a unique way. \emph{Progress}: final states have the form $\statesearch{ \la\var\code } \stempty$ and $\statesearch \var \stack$, that both decode to $\towh$-normal forms.
\end{proof}

\begin{lemma}[One-Step Simulation]
\label{l:searching-am-one-step-reverse}
Let $\state$ be a Searching AM state. If $\decode\state \towh \tmtwo$ then there exists a state $\statetwo$ such that $\state \tomachappl^* \tomachrbeta \statetwo$ and $\decode\statetwo = \tmtwo$.
\end{lemma}

\begin{proof}
Let $\nfpar\state\applsym$ be the normal form of $\state$ with respect to $\tomachappl$, that exists and is unique by termination of $\tomachappl$ (\reflemmap{searching-am-bricks}{termination}) and determinism of the machine (\reflemmap{searching-am-bricks}{determinism}). Since $\tomachappl$ is mapped on identities (\reflemmap{searching-am-one-step}{commutative}) one has $\decode{ \nfpar\state\applsym } = \decode\state$. By hypothesis $\decode\state$ $\towh$-reduces, so that by progress (\reflemmap{searching-am-bricks}{progress}) $\nfpar\state\applsym$ cannot be final. Then $\nfpar\state\applsym \tomachrbeta \statetwo$, and $\decode{ \nfpar\state\applsym } = \decode\state \towh \decode\statetwo$ by the one-step simulation lemma (\reflemmap{searching-am-one-step}{principal}). By determinism of $\towh$, one obtains $\decode\statetwo = \tmtwo$.
\end{proof}

Finally, we obtain the implementation theorem.

\begin{theorem}
The Searching AM implements the weak head strategy.
\end{theorem}

\begin{proof}
\emph{Executions to Derivations}: by induction on the length $\size\exec$ of $\exec$ using \reflemma{searching-am-one-step}. \emph{Derivations to Executions}: by induction on the length $\size\deriv$ of $\deriv$ using \reflemma{searching-am-one-step-reverse} and noting that $\decode{\compil\tm} = \tm$.
\end{proof}

\begin{figure}
   \begin{center}
   \begin{tabular}{c}
   \begin{tabular}{c|c}
   $\begin{array}{rrclc}
     \mbox{Environments} & \genv & \defeq & \stempty \mid \esub\var\code \cons \genv\\
     \mbox{Compilation} & \compil\tm & \defeq & \statesearch \code \stempty\\
         \end{array}$
   &
   $\begin{array}{rrclc}
     \mbox{Decoding} & \decode{ \statesearch \code \stempty } & \defeq & \tm\\
      & \decode{ \statesearch \code {\esub\var\codetwo \cons \genv} } & \defeq & \decode{ \statesearch {\code \isub\var\codetwo} \genv }
         \end{array}$
         
         \end{tabular}
\\\\
   $
   \begin{array}{|c|c||c||c|c|}
   \hline
		Code & Env & Trans & Code & Env
\\
		\hhline{=|=||=||=|=}								
		(\la{\var}\code)\codetwo \codethree_1 \ldots \codethree_k & \genv
	  	&\tomachhole{d\beta} &
		\code \codethree_1 \ldots \codethree_k & \esub{\var}{\codetwo}\cons\genv
		\\
		\hhline{-|-||-||-|-}
		\var \codethree_1 \ldots \codethree_k & \genv\cons \esub{\var}{\code} \cons\genvtwo
	  	&\tomachhole{var}&
	  	{\rename{\code}} \codethree_1 \ldots \codethree_k  & \genv\cons \esub{\var}{\code} \cons\genvtwo\\
\hline
  	\end{array}
	$
   \end{tabular}   
  	
	where ${\rename{\code}}$ denotes $\code$ where bound names have been freshly renamed.

  \end{center}

  \caption{Micro-Substituting Abstract Machine (Micro AM).}
     \label{fig:micro-am}
\end{figure}

\section{The Micro-Substituting Abstract Machine}

\paragraph{Decomposing Meta-Level Substitution.} The second task of abstract machines is to replace meta-level substitution $\code\isub\var\codetwo$ with \emph{micro-step substitution on demand}, \ie a parsimonious approximation of meta-level substitution based on:
\begin{enumerate}
\item \emph{Sharing}: when a $\beta$-redex $(\la\var\code)\codetwo$ is in evaluation position it is fired but the meta-level substitution $\code\isub\var\codetwo$ is delayed, by introducing an annotation $\esub\var\codetwo$ in a data-structure for delayed substitutions called \emph{environment};
\item \emph{Micro-Step Substitution}: variable occurrences are replaced one at a time;
\item \emph{Substitution on Demand}: replacement of a variable occurrence happens only when it ends up in evaluation position---variable occurrences that do not end in evaluation position are never substituted.
\end{enumerate}
The purpose of this section is to illustrate this process in isolation via the study of a toy machine, the \emph{Micro-Substituting Abstract Machine} (Micro AM) in \reffig{micro-am}, forgetting about the search for redexes.

\paragraph{Environments.} We are going to treat environments in an unusual way: the literature mostly deals with \emph{local} environments, to be discussed in \refsect{names}, while here we prefer to first address the simpler notion of \emph{global} environment, but to ease the terminology we will simply call them \emph{environments}. So, an \emph{environment} $\genv$ is a list of entries of the form $\esub\var\codetwo$. Each entry denotes the \emph{delayed} substitution of $\codetwo$ for $\var$. In a state $\statemicro\code {\genvtwo\cons\esub\var\codetwo\cons\genvthree}$ the scope of $\var$ is given by $\code$ and $\genvtwo$, as it is stated by forthcoming \reflemma{micro-am-name-invariant}. The (global) environment models a store. As it is standard in the literature, it is a \emph{list}, but the list structure is only used to obtain a simple decoding and a handy delimitation of the scope of its entries. These properties are useful to develop the meta-theory of abstract machines, but keep in mind that (global) environments are not meant to be implemented as lists.

 \paragraph{Code.} The code under evaluation is now a $\l$-term $h \codethree_1 \ldots \codethree_k$ expressed as a head $h$ (that is either a $\beta$-redex $(\la{\var}\code)\codetwo$ or a variable $\var$) applied to $k$ arguments---it is a by-product of the fact that the Micro AM does not address \emph{search}. 
 
 \paragraph{Transitions.} There are two transitions:
 \begin{itemize}
 	\item \emph{Delaying $\beta$}: transition $\tomachhole{d\beta}$ removes the $\beta$-redex $(\la{\var}\code)\codetwo$ but does not execute the expected substitution $\isub\var\codetwo$, it rather delays it, adding $\esub\var\codetwo$ to the environment. It is the $\beta$-transition of the Micro AM.
	\item \emph{Micro-Substitution On Demand}: if the head of the code is a variable $\var$ and there is an entry $\esub\var\code$ in the environment then transition $\tomachvar$ replaces that occurrence of $\var$---and only that occurrence---with a copy of $\code$. It is necessary to rename the new copy of $\code$ (into a well-named term) to avoid name clashes. It is the overhead transition of the Micro AM.
 \end{itemize}
 
 \paragraph{Implementation.} Compilation sends a (well-named) term $\tm$ to the initial state $\statemicro\code\stempty$, as for the Searching AM (but now the empty data-structure is the environment). The decoding simply applies the delayed substitutions in the environment to the term, considering them as meta-level substitutions.
 
 The implementation of weak head reduction $\towh$ by the Micro AM can be shown using the recipe given for the Searching AM, and it is therefore omitted. The only difference is in the proof that the overhead transition $\tomachvar$ terminates, that is based on a different argument. We spell it out because it will be useful also later on for complexity analyses. It requires the following invariant of machine executions:
 
 \begin{lemma}[Name Invariant]
 \label{l:micro-am-name-invariant} 
  Let $\state = \statemicro \code \genv$ be a Micro AM reachable state. 
  \begin{enumerate}
    \item \label{p:micro-am-name-invariant-abs} \emph{Abstractions}: if $\la\var\codetwo$ is a subterm of $\code$ or of any code in $\genv$ then $\var$ may occur only in $\codetwo$;
    \item \label{p:micro-am-name-invariant-env} \emph{Environment}: if $\genv = \genvtwo \cons \esub\var\codetwo \cons \genvthree$ then $\var$ is fresh with respect to $\codetwo$ and $\genvthree$.    
  \end{enumerate}   
 \end{lemma}
 
 \begin{proof}
   By induction on the length of the execution $\exec$ leading to $\state$. If $\exec$ is empty then $\state$ is initial and the statement holds because $\code$ is well-named by hypothesis. If $\exec$ is non-empty then it follows from the \ih and the fact that transitions preserve the invariant, as an immediate inspection shows.
 \end{proof}

\begin{lemma}[Micro-Substitution Terminates]
\label{l:exp-trans-terminate}
  $\tomachvar$ terminates in at most $\size\genv$ steps (on reachable states). 
\end{lemma}

\begin{proof}
  Consider a $\tomachvar$ transition copying $\codetwo$ from the environment $\genvtwo \cons \esub\var\codetwo \cons \genvthree$. If the next transition is again $\tomachvar$, then the head of $\codetwo$ is a variable $\vartwo$ and the transition copies from an entry in $\genvthree$ because by \reflemma{micro-am-name-invariant} $\vartwo$ cannot be bound by the entries in $\genvtwo$. Then the number of consecutive $\tomachvar$ transitions is bound by $\genv$ (that is not extended by $\tomachvar$).
\end{proof}

\begin{theorem}
The Micro AM implements the weak head strategy.
\end{theorem}

\section{Search + Micro-Substitution = Milner Abstract Machine}
The Searching AM and the Micro AM can be merged together into the Milner Abstract Machine (MAM), defined in \reffig{mam}. The MAM has both an argument stack and an environment. The machine has one $\beta$-transition $\tomachrbeta$ inherited from the Searching AM, and two overhead transitions, $\tomachappl$ inherited from the the Searching AM and $\tomachvar$ inherited from the Micro AM. Note that in $\tomachvar$ the code now is simply a variable, because the arguments are supposed to be stored in the argument stack.

For the implementation theorem once again the only delicate point is to prove that the overhead transitions terminate. As for the Micro AM one needs a name invariant. A termination measure can then be defined easily by mixing the size of the codes (needed for $\tomachappl$) and the size of the environment (needed for $\tomachvar$), and it is omitted here, because it will be exhaustively studied for the complexity analysis of the MAM. Therefore, we obtain that:

\begin{theorem}
The MAM implements the weak head strategy.
\end{theorem}

\begin{figure}
\begin{center}
  \begin{tabular}{c}
    \begin{tabular}{c|c}
   $\begin{array}{rrclc}
     \mbox{Environments} & \genv & \defeq & \stempty \mid \esub\var\code \cons \genv\\
      \mbox{Stacks} & \stack & \defeq & \stempty \mid \code \cons \stack\\
     \mbox{Compilation} & \compil\tm & \defeq & \statemam \code \stempty \stempty \\
         \end{array}$
         &
         $\begin{array}{rrclc}
     \mbox{Decoding} & \decode{ \statemam \code \stempty \stempty } & \defeq & \tm\\
      & \decode{ \statemam \code {\codetwo \cons \stack} \genv } & \defeq & \decode{ \statemam {\code \codetwo} \stack \genv}\\
      & \decode{ \statemam \code \stempty {\esub\var\codetwo \cons \genv} } & \defeq & \decode{ \statemam {\code \isub\var\codetwo} \stempty \genv }
         \end{array}$
   
   \end{tabular}
   
\\\\
  	$\begin{array}{|c|c|c||c||c|c|c|cccc}
   \hline
		Code & Stack & Env & Trans & Code & Stack & Env
\\
		\hhline{=|=|=||=||=|=|=}			
				\code\codetwo & \stack & \genv
	  	&\tomachhole{@l}&
	  	\code & \codetwo\cons\stack & \genv
		\\
		\hhline{-|-|-||-||-|-|-}
						
		\la{\var}\code & \codetwo\cons\stack & \genv
	  	&\tomachhole{r\beta} &
		\code & \stack & \esub{\var}{\codetwo}\cons\genv
		\\
		\hhline{-|-|-||-||-|-|-}
		\var & \stack & \genv\cons \esub{\var}{\code} \cons\genvtwo
	  	&\tomachvar&
	  	{\rename{\code}}  & \stack & \genv\cons \esub{\var}{\code} \cons\genvtwo\\
		\hline
  	\end{array}$
  	   \end{tabular}
  	   
  	   where ${\rename{\code}}$ denotes $\code$ where bound names have been freshly renamed.
	\end{center}
	\caption{Milner Abstract Machine (MAM).}
		\label{fig:mam}
\end{figure}

\section{Introducing Complexity Analyses}
\label{sect:intro-complexity}
The complexity analysis of abstract machines is the study of the asymptotic behavior of their overhead. 

\paragraph{Parameters for Complexity Analyses.} Let us reason abstractly, by considering a generic strategy $\tostrat$ in the $\l$-calculus and a given machine $\mach$ implementing $\tostrat$. By the \emph{derivations-to-executions} part of the implementation (\refdef{implem}), given a derivation $\deriv: \tm_0 \tostrat^n \tmtwo$ there is a shortest execution $\exec: \compil{\tm_0} \tomachine \state$ such that $\decode\state = \tmtwo$. Determining \emph{the complexity of $\mach$} amounts to bound the complexity of a concrete implementation of $\exec$, say on a RAM model, as a function of two fundamental parameters:
\begin{enumerate}
  \item \emph{Input}: the size $\size{\tm_0}$ of the initial term $\tm_0$ of the derivation $\deriv$;
  \item \emph{Strategy} the length $n = \size\deriv$ of the derivation $\deriv$, that coincides with the number $\sizebeta{\exec}$ of $\beta$-transitions in $\exec$  by the $\beta$-matching requirement for implementations.
\end{enumerate}
Note that our notion of implementation allows to forget about the strategy while studying the complexity of the machine, because the two fundamental parameters are internalized: the input is simply the initial code and the length of the strategy is simply the number of $\beta$-transitions.

\paragraph{Types of Machines.} The bound on the overhead of the machine is then used to classify it, as follows.

\begin{definition}
Let $\mach$ an abstract machine implementing a strategy $\tostrat$. Then
\begin{itemize}
  \item $\mach$ is \emph{reasonable} if the complexity of $\mach$ is polynomial in the input $\size{\tm_0}$ and the strategy $\sizebeta{\exec}$;
  \item $\mach$ is \emph{unreasonable} if it is not reasonable;
  \item $\mach$ is \emph{efficient} if it is linear in both the input and the strategy (we sometimes say that it is \emph{bilinear}).
\end{itemize}
\end{definition}

\paragraph{Recipe for Complexity Analyses.} The estimation of the complexity of a machine usually takes 3 steps:

\begin{enumerate}
  \item \emph{Number of Transitions}: bound the length of the execution $\exec$ simulating the derivation $\deriv$, usually having a bound on every kind of transition of $\mach$.
  \item \emph{Cost of Single Transitions}: bound the cost of concretely implementing a single transition of $\mach$---different kind of transitions usually have different costs. Here it is usually necessary to go beyond the abstract level, making some (high-level) assumption on how codes and data-structure are concretely represented (our case study will provide examples).
  \item \emph{Complexity of the Overhead}: obtain the total bound by composing the first two points, that is, by taking the number of each kind of transition times the cost of implementing it, and summing over all kinds of transitions.
\end{enumerate}
\section{The Complexity of the MAM}

In this section we provide the complexity analysis of the MAM, from which analyses of the Searching and Micro AM easily follow. 

\paragraph{The Crucial Subterm Invariant.} The analysis is based on the following subterm invariant. 

\begin{lemma}[Subterm Invariant]
\label{l:subterm}
Let $\exec: \compil{\tm_0} \tomachhole{MAM} \statemam \codetwo \stack \genv$ be a MAM execution. Then $\codetwo$ and any code in $\stack$ and $\genv$ are subterms of $\tm_0$.
\end{lemma}

Note that the MAM copies code only in transition $\tomachvar$, where it copies a code from the environment $\genv$. Therefore, the subterm invariant bounds the size of the subterms duplicated along the execution.

Let us be precise about \emph{subterms}: for us, $\codetwo$ is a subterm of $\tm_0$ if it does so up to variable names, both free and bound (and so the distinction between terms and codes is irrelevant). More precisely: define $\tm^-$ as $\tm$ in which all variables (including those appearing in binders) are replaced by a fixed symbol $\ast$. Then, we will consider $\tmtwo$ to be a subterm of $\tm$ whenever $\tmtwo^-$ is a subterm of $\tm^-$ in the usual sense. The key property ensured by this definition is that the size $\size\codetwo$ of $\codetwo$ is bounded by $\size\code$.

\begin{proof}
  By induction on the length of $\exec$. The base case is immediate and the inductive one follows from the \ih and the immediate fact that the transitions preserve the invariant.
\end{proof}

The subterm invariant is crucial, for two related reasons. First, it linearly relates the cost of duplications to the size of the input, enabling complexity analyses. With respect to the length of the strategy, then, micro-step operations have constant cost, as required by the recipe for micro-step operational semantics in \refsect{indeed}. Second, it implies that size-explosion has been circumvented: duplications are linear, and so the size of the state can grow at most linearly with the number of steps, \ie it cannot explode. In particular, we also obtain the compact representation of the results required by the recipe.

The relevance of the subterm invariant goes in fact well beyond abstract machines, as it is typical of most instances of  micro-step operational semantics. And for complexity analyses of the $\l$-calculus it is absolutely essential, playing a role analogous to that of the cut-elimination theorem in the study of sequent calculi or of the sub-formula property for proof search.

\paragraph{Number of Transitions.} The next lemma bounds the global number of overhead transitions. For the micro-substituting transition $\tomachvar$ it relies on an auxiliary bound of a more local form. For the searching transition $\tomachappl$ the bound relies on the subterm invariant. We denote with $\sizebeta\exec$, $\sizeappl\exec$, and $\sizevar\exec$ the number of $\tomachrbeta$, $\tomachappl$, and $\tomachvar$ transitions in $\exec$, respectively.

\begin{lemma}
\label{l:mam-trans-bound} 
  Let $\exec: \compil{\tm_0} \tomachhole{MAM} \state$ be a MAM execution. Then:
  \begin{enumerate}
    \item \label{p:mam-trans-bound-exp-local} \emph{Micro-Substitution Linear Local Bound}: if $\exectwo: \state \tomachhole{\applsym,\varsym}^* \statetwo$ then $\sizevar\exectwo \leq \size\genv = \sizebeta\exec$;
    \item \label{p:mam-trans-bound-exp-global} \emph{Micro-Substitution Quadratic Global Bound}: $\sizevar\exec \leq \sizebeta\exec^2$;
    \item \label{p:mam-trans-bound-com-local} \emph{Searching (and $\beta$) Local Bound}: if $\exectwo: \state \tomachhole{\rootbetasym,\applsym}^* \statetwo$ then $\size\exectwo \leq \size{\tm_0}$;
    \item \label{p:mam-trans-bound-com-global} \emph{Searching Global Bound}: $\sizeappl\exec \leq \size{\tm_0} \cdot (\sizevar\exec +1) \leq \size{\tm_0} \cdot (\sizebeta\exec^2 + 1)$.
  \end{enumerate}
\end{lemma}

\proof\hfill
  \begin{enumerate}
     \item Reasoning along the lines of \reflemma{exp-trans-terminate} one obtains that $\tomachvar$ transitions in $\exectwo$ have to use entries of $\genv$ from left to right ($\tomachappl$ and $\tomachvar$ do not modify $\genv$), and so  $\sizevar\exectwo \leq \size\genv$. Now, $\size\genv$ is exactly $\sizebeta\exec$, because the only transition extending $\genv$, and of exactly one entry, is $\tomachrbeta$. 
     
     \item The fact that a linear local bound induces a quadratic global bound is a standard reasoning. We spell it out to help the unacquainted reader. The execution $\exec$ alternates phases of $\beta$-transitions and phases of overhead transitions, \ie it has the shape:
  $$\compil{\tm_0} = \state_1 \tomachrbeta^* \statetwo_1 \tomachhole{\applsym,\varsym}^* 
		     \state_2 \tomachrbeta^* \statetwo_2 \tomachhole{\applsym,\varsym}^*   
	\ldots	     \state_k \tomachrbeta^* \statetwo_k \tomachhole{\applsym,\varsym}^* \state$$
Let $a_i$ be the length of the segment $\state_i \tomachrbeta^* \statetwo_i$ and $b_i$ be the number of $\tomachvar$ transitions in the segment $\statetwo_i \tomachhole{\applsym,\varsym}^* \state_{i+1}$, for $i = 1,\ldots, k$.
  By \refpoint{mam-trans-bound-exp-local}, we obtain $b_i\leq \sum_{j=1}^i
  a_j$. Then $\sizevar\exec=\sum_{i=1}^k b_i\leq\sum_{i=1}^k \sum_{j=1}^i a_j$. 
  Note that $\sum_{j=1}^i a_j\leq \sum_{j=1}^k a_j=\sizebeta\exec$ and
  $k \leq \sizebeta\exec$. So
$\sizevar\exec\leq \sum_{i=1}^k\sum_{j=1}^i a_j \leq \sum_{i=1}^k\sizebeta\exec \leq \sizebeta\exec^2$.

\item 	The length of $\exectwo$ is bound by the size of the code in the state $\state$ because $\tomachhole{\rootbetasym,\applsym}$ strictly decreases the size of the code, that in turn is bound by the size $\size{\tm_0}$ of the initial term by the subterm invariant (\reflemma{subterm}). 

    \item The execution $\exec$ alternates phases of $\tomachrbeta$ and $\tomachappl$ transitions and phases of $\tomachvar$ transitions,  \ie it has the shape:
  $$\compil{\tm_0} = \state_1 \tomachhole{\rootbetasym,\applsym}^* \statetwo_1 \tomachvar^* 
		     \state_2 \tomachhole{\rootbetasym,\applsym}^* \statetwo_2 \tomachvar^*   
	\ldots	     \state_k \tomachhole{\rootbetasym,\applsym}^* \statetwo_k \tomachvar^* \tomachhole{\rootbetasym,\applsym}^* \state$$
	
By \refpoint{mam-trans-bound-com-local} the length of the segments $\state_i \tomachhole{\rootbetasym,\applsym}^* \statetwo_i$ is bound by the size $\size{\tm_0}$ of the initial term. The code may grow, instead, with $\tomachvar$ transitions. So $\sizeappl\exec$ is bound by $\size{\tm_0}$ times the number $\sizevar\exec$ of micro-substitution transitions, plus $\size{\tm_0}$ once more, because at the beginning there might be $\tomachhole{\rootbetasym,\applsym}$ transitions before any $\tomachvar$ transition---in symbols, $\sizeappl\exec \leq \size{\tm_0} \cdot (\sizevar\exec +1)$. Finally, $\size{\tm_0} \cdot (\sizevar\exec +1) \leq \size{\tm_0} \cdot (\sizebeta\exec^2 + 1)$ by \refpoint{mam-trans-bound-exp-global}.\qed
  \end{enumerate}
  
  \paragraph{Cost of Single Transitions.} To estimate the cost of concretely implementing single transitions we need to make some hypotheses on how the MAM is going to be itself implemented on RAM:
  
  \begin{enumerate}  
    \item \emph{Codes, Variable (Occurrences), and Environment Entries}: abstractions and applications are constructors with pointers to subterms, a variable is a memory location, a variable occurrence is a reference to that location, and an environment entry $\esub\var\code$ is the fact that the location associated to $\var$ contains (the topmost constructor of) $\code$.    
    \item \emph{Random Access to Global Environments}: the environment $\genv$ of the MAM can be accessed in constant time (in $\tomachvar$) by just following the reference given by the variable occurrence under evaluation, with no need to access $\genv$ sequentially, thus ignoring its list structure.
  \end{enumerate}
  
  It is now possible to bound the cost of single transitions. Note that the case of $\tomachvar$ transitions relies on the subterm invariant.
  
\begin{lemma}
\label{l:mam-single-trans-cost}
  Let $\exec: \compil{\tm_0} \tomachhole{MAM} \state$ be a MAM execution. Then:
  \begin{enumerate}
    \item Each $\tomachappl$ transition in $\exec$ has constant cost;
    \item Each $\tomachrbeta$ transition in $\exec$ has constant cost;
    \item Each $\tomachvar$ transition in $\exec$ has cost bounded by the size $\size{\tm_0}$ of the initial term.
  \end{enumerate}
\end{lemma}

\begin{proof}
  According to our hypothesis on the concrete implementation of the MAM, $\tomachappl$ just moves the pointer to the current code on the left subterm of the application and pushes the pointer to the right subterm on the stack---evidently constant time. Similarly for $\tomachrbeta$. For $\tomachvar$, the environment entry $\esub\var\code$ is accessed in constant time by hypothesis, but $\code$ has to be $\alpha$-renamed, \ie copied. It is not hard to see that this can be done in time linear in $\size\code$ (the naive algorithm for copying carries around a list of variables, and it is quadratic, but it can be easily improved to be linear) that by the subterm invariant (\reflemma{subterm}) is bound by the size $\size{\tm_0}$ of the initial term.
\end{proof}

\paragraph{Complexity of the Overhead.} By composing the analysis of the number of transitions (\reflemma{mam-trans-bound}) with the analysis of the cost of single transitions (\reflemma{mam-single-trans-cost}) we obtain the main result of the paper.

\begin{theorem}[The MAM is Reasonable]
\label{thm:mam-reasonable} 
  Let $\exec: \compil{\tm_0} \tomachhole{MAM} \state$ be a MAM execution. Then:
  \begin{enumerate}
    \item \label{p:mam-reasonable-com} $\tomachappl$ transitions in $\exec$ cost all together $O(\size{\tm_0} \cdot (\sizebeta\exec^2 + 1))$;
    \item $\tomachrbeta$ transitions in $\exec$ cost all together $O(\sizebeta\exec)$;
    \item $\tomachvar$ transitions in $\exec$ cost all together $O(\size{\tm_0} \cdot (\sizebeta\exec^2 + 1))$;
  \end{enumerate}
  Then $\exec$ can be implemented on RAM with cost $O(\size{\tm_0} \cdot (\sizebeta\exec^2 + 1))$, \ie the MAM is a reasonable implementation of the weak head strategy.
\end{theorem}

\paragraph{The Efficient MAM.} According to the terminology of \refsect{machines-intro}, the MAM is reasonable but it is not efficient because micro-substitution takes time quadratic in the length of the strategy. The quadratic factor comes from the fact that in the environment there can be growing chains of renamings, \ie of substitutions of variables for variables, see \cite{DBLP:conf/wollic/AccattoliC14} for more details on this issue. The MAM can actually be optimized easily, obtaining an efficient implementation, by replacing $\tomachrbeta$ with the following two $\beta$-transitions:
\begin{center}
$\begin{array}{|c|c|c||c||c|c|c|cccc}
		\hhline{-|-|-|-|-|-|-}
		\la{\var}\code & \vartwo\cons\stack & \genv
	  	&\tomachhole{r\beta_1} &
		\code\isub\var\vartwo & \stack & \genv
		\\
		\hhline{-|-|-||-||-|-|-}
		\la{\var}\code & \codetwo\cons\stack & \genv
	  	&\tomachhole{r\beta_2} &
		\code & \stack & \esub{\var}{\codetwo}\cons\genv &\mbox{if $\codetwo$ is not a variable}\\
		\hhline{-|-|-|-|-|-|-}
  	\end{array}$
\end{center}

\paragraph{Search is Linear and the Micro AM is Reasonable.} By \reflemma{mam-trans-bound} the cost of search in the MAM is linear in the number of transitions for implementing micro-substitution. This is an instance of a more general fact: \emph{search} turns out to always be bilinear (in the initial code and in the amount of micro-substitutions). There are two consequences of this general fact. First, it can be turned into a design principle for abstract machines---search \emph{has to be bilinear}, otherwise there is something wrong in the design of the machine. Second, search is somewhat negligible for complexity analyses. 

The Micro AM can be seen as the MAM up to search. In particular, it satisfies a subterm invariant and thus circumvents size-explosion. The Micro AM is however quite less efficient, because at each step it has to search the redex from scratch. An easy but omitted analisys shows that its overhead is nonetheless polynomial. Therefore, it makes sense to consider very abstract machines as the Micro AM that omit search. In fact, they already exist, in disguise, as strategies in the \emph{linear substitution calculus} \cite{DBLP:conf/rta/Accattoli12,DBLP:conf/popl/AccattoliBKL14}, a recent approach to explicit substitutions modeling exactly micro-substitution without search (the traditional approach to explicit substitutions instead models both micro-substitution and search) and they were used for the first proof that a strong strategy (the leftmost-outermost one) is reasonable \cite{DBLP:conf/csl/AccattoliL14}.

\paragraph{The Searching AM is Unreasonable.} It is not hard to see that the Searching AM is unreasonable. Actually, the number of transitions is reasonable. The projection of MAM executions on Searching AM executions, indeed, shows that the number of searching transitions of the Searching AM is reasonable. It is the cost of single transitions that becomes unreasonable. In fact, the Searching AM does not have a subterm invariant, because it rests on meta-level substitution, and the size of the terms duplicated by the $\tomachrbeta$ transition can explode (it is enough to consider the size-exploding family of \refprop{abs-size-explosion}).

The moral is that micro-substitution is more fundamental than search. While the cost of search can be expressed in terms of the cost of micro-substitution, the converse is in fact not possible.

\section{Names: Krivine Abstract Machine}
\label{sect:names}
\paragraph{Accounting for Names.} In the study presented so far we repeatedly took names seriously, by distinguishing between terms and  codes, by asking that initial codes are well-named, and by proving an invariant about names (\reflemma{micro-am-name-invariant}). The process of $\alpha$-renaming however has not been made explicit, the machines we presented rather rely on a meta-level renaming, used as a black box.

The majority of the literature on abstract machines, instead, pays more attention to $\alpha$-equivalence, or rather to how to avoid it. We distinguish two levels:
\begin{enumerate}
  \item \emph{Removal of on-the-fly $\alpha$-equivalence}: in these cases the machine works on terms with variable names but it is designed in order to implement evaluation without ever $\alpha$-renaming. Technically, the global environment of the MAM is replaced by many local environments, each one for every piece of code in the machine. The machine becomes more complex, in particular the non-trivial concept of closure (to be introduced shortly) is necessary.
  \item \emph{Removal of names}: terms are represented using de Bruijn indexes (or de Bruijn levels), removing the problem of $\alpha$-equivalence altogether but sacrificing the readability of the machine and reducing its abstract character. Usually this level is built on top of the previous one.
\end{enumerate}
We are now going to introduce Krivine Abstract Machine (keeping names, so at the first level), yet another implementation of the weak head strategy. Essentially, it is a version of the MAM without on-the-fly $\alpha$-equivalence. The complexity analysis will show that it has exactly the same complexity of the MAM. The further removal of names is only (anti)cosmetic---the complexity is not affected either. Consequently, the task of accounting for names is---as for search---negligible for complexity analyses.

  \begin{figure}
  \begin{center}
\begin{tabular}{c}
\begin{tabular}{c|c}
   $\begin{array}{rrclc}
     \mbox{Local Env.} & \env & \defeq & \stempty \mid \esub\var\clos \cons \env\\
     \mbox{Closures} & \env & \defeq & (\code, \env)\\
      \mbox{Stacks} & \stack & \defeq & \stempty \mid \clos \cons \stack\\
      \mbox{States} & \state & \defeq & \statekamclos \clos \stack \\
     \mbox{Compilation} & \compil\tm & \defeq & \statekamclos {(\code, \stempty)} \stempty \\
     
         \end{array}$
  &
  $\begin{array}{rrclc}
     
     \mbox{Closure Decoding} & \decode{ (\code, \stempty) } & \defeq & \tm\\
      & \decode{ (\code, \esub\var\clos \cons \env ) } & \defeq & \decode{ (\code \isub\var{\decode\clos}, \env ) }\\      
      \mbox{State Decoding} & \decode{ \statekamclos \clos \stempty } & \defeq & \decode\clos\\
      & \decode{ \statekamclos \clos {\clostwo \cons \stack}  } & \defeq & \decode{ \statekamclos {(\decode\clos \decode\clostwo, \stempty)} \stack}  
         \end{array}$
    \end{tabular}
\\\\
  	$\begin{array}{|c|c|c||c||c|c|c|cccc}
   \hhline{-|-|-|-|-|-|-}
		Code & Loc Env & Stack &  Trans & Code & Loc Env & Stack
\\
		\hhline{=|=|=||=||=|=|=}			
				\code\codetwo&\env&\stack
				&\tomachappl&
				\code&\env&(\codetwo,\env)\cons\stack
				\\ \hhline{-|-|-||-||-|-|-}
				\l\var.\code&\env&\clos\cons\stack
				&\tomachrbeta&
				\code&\esub\var\clos\cons\env&\stack
		\\ \hhline{-|-|-||-||-|-|-}
				\var&\env&\stack
				&\tomachvar&
				\code&\envtwo&\stack &\mbox{with $\env(\var)=(\code,\envtwo)$}
		\\				
		\hhline{-|-|-|-|-|-|-}
  	\end{array}$
  	   \end{tabular}
  \end{center}
  \caption{Krivine Abstract Machine (KAM).}
  \label{fig:kam}
\end{figure}

\paragraph{Krivine Abstract Machine.} The machine is in \reffig{kam}. It relies on the mutually inductively defined concepts of \emph{local environment}, that is a list of closures, and \emph{closure}, that is a list of pairs of a code and a local environment. A state is a \emph{pair} of a closure and a stack, but in the description of the transitions we write it as a \emph{triple}, by spelling out the two components of the closure. Let us explain the name \emph{closure}: usually, machines are executed on closed terms, and then a closure decodes indeed to a closed term. While essential in the study of call-by-value or call-by-need strategies, for the weak head (call-by-name) strategy the closed hypothesis is unnecessary, that is why we do not deal with it---so a closure here does not necessarily decode to a closed term. Two remarks:
\begin{enumerate}
  \item \emph{Garbage Collection}: transition $\tomachvar$, beyond implementing micro-substitution, also accounts for some garbage collection, as it throws away the local environment $\env$ associated to the replaced variable $\var$. The MAM simply ignores garbage collection. For time analyses garbage collection can indeed be safely ignored, while it is clearly essential for space (both the KAM and the MAM are however desperately inefficient with respect to space).
  
  \item \emph{No $\alpha$-Renaming and the Length of Local Environments}: names are never renamed. The initial code, as usual, is assumed to be well-named. Then the entries of a same local environment are all on distinguished names (formally, a name invariant holds). Then the length of a local environment $\env$ is bound by the number of names in the initial term, that is, by the size of the initial term (formally, $\size\env \leq \size{\tm_0}$). This  essential quantitative invariant is used in analisys of the next paragraph.
\end{enumerate}

\paragraph{Implementation and Complexity Analysis.} The proof that the KAM implements the weak head strategy follows the recipe for these proofs and it is omitted. For the complexity analysis, the bound of the number of transitions can be shown to be exactly as for the MAM. A direct proof is not so simple, because the bound on $\tomachvar$ transitions cannot exploit the size of the global environment. The bound can be obtained by relating the KAM with the Searching AM (for which the exact same bound of the MAM holds), or by considering the \emph{depth} (\ie maximum nesting) of local environments. The proof is omitted. 

The interesting part of the analysis is rather the study of the cost of single transitions. As for the MAM, we need to spell out the hypotheses on how the KAM is concretely implemented on RAM. Variables cannot be implemented with pointers, because the same variable name can be associated to different codes in different local environments. So they have to simply be numbers. Then there are two choices for the representation of environments, either they are represented as lists or as arrays. In both cases $\tomachrbeta$ can be implemented in constant time. For the other transitions:
\begin{enumerate}
  \item \emph{Environments as Arrays}: we mentioned that there is a bound on the length of local environments ($\size\env \leq \size{\tm_0}$) so that arrays can be used. The choice allows to implement $\tomachvar$ in constant time, because $\env$ can be accessed directly at the position described by the number given by $\var$. Transition $\tomachappl$ however requires to duplicate $\env$, and this is necessary because the two copies might later on be modified differently. So the cost of a $\tomachappl$ transition becomes linear in $\size{\tm_0}$ and $\tomachappl$ transitions all together cost $O(\size{\tm_0}^2 \cdot (\sizebeta\exec^2 + 1))$, that also becomes the complexity of the whole overhead of the KAM. This is worse than the MAM.
  
  \item \emph{Environments as Lists}: implementing local environments as lists provides sharing of environments, overcoming the problems of arrays. With lists, transition $\tomachappl$ becomes constant time, as for the MAM, because the copy of $\env$ now is simply the copy of a pointer. The trick is that the two copies of the environment can only be extended differently \emph{on the head}, so that the tail of the list can be shared. Transition $\tomachvar$ however now needs to access $\env$ sequentially, and so it costs $\size{\tm_0}$ as for the MAM. Thus globally we obtain the same overhead of the MAM.
\end{enumerate}

Summing up, \emph{names} can be pushed at the meta-level (as in the MAM) without affecting the complexity of the overhead. Thus, \emph{names} are even less relevant than \emph{search} at the level of complexity. The moral of this tutorial then is that \emph{substitution} is the crucial aspect for the complexity of abstract machines.

\bibliographystyle{eptcs}
\bibliography{\macrospath/biblio}
\end{document}